\documentclass[aps,prl,letterpaper,superscriptaddress,twocolumn]{revtex4}

\usepackage{amstext,amsmath,amssymb,amsthm}
\usepackage{color}

\usepackage{times}

\usepackage{graphicx}
\usepackage[colorlinks]{hyperref}

\theoremstyle{plain}

\newtheorem{theorem}{Theorem}

\theoremstyle{definition}
\newtheorem{definition}{Definition}

\newcommand{\ket} [1] {\vert #1 \rangle}

\newcommand{\ba}{\begin{align}}
\newcommand{\ea}{\end{align}}
\newcommand{\bea}{\begin{eqnarray}}
\newcommand{\eea}{\end{eqnarray}}

\definecolor{ao}{rgb}{0.0, 0.5, 0.0}

 \definecolor{CYAN}{cmyk}{1,0,0,0}
 \definecolor{MAGENTA}{cmyk}{0,1,0,0}
 \definecolor{YELLOW}{cmyk}{0,0,1,0}

\def\id{I}

\def\1{\mat{\id}}
\def\mat#1{\mathbf{#1}}

\def\EQ#1{\begin{eqnarray}#1\end{eqnarray}}

\begin{document} 

\title{Smooth input preparation for quantum and quantum-inspired machine learning}

\author{Zhikuan Zhao}
\email{zhikuan.zhao@inf.ethz.ch}
\affiliation{Department of Computer Science, ETH Zurich,  Universitatstrasse 6, 8092 Zurich}
\affiliation{Singapore University of Technology and Design, 8 Somapah Road, Singapore 487372}
\affiliation{Centre for Quantum Technologies, National University of Singapore, 3 Science Drive 2, Singapore 117543}
\author{Jack K. Fitzsimons} 
\affiliation{Department of Engineering Science, University of Oxford, Oxford OX1 3PJ, UK}
\author{Patrick Rebentrost}
\affiliation{Centre for Quantum Technologies, National University of Singapore, 3 Science Drive 2, Singapore 117543}
\author{Vedran Dunjko}
\affiliation{Max-Planck-Institut f\"{u}r Quantenoptik, Hans-Kopfermann-Str. 1, 85748 Garching, Germany}
\affiliation{LIACS, Leiden University, Niels Bohrweg 1, 2333 CA Leiden, The Netherlands}
\author{Joseph F. Fitzsimons}
\email{joseph_fitzsimons@sutd.edu.sg}
\affiliation{Singapore University of Technology and Design, 8 Somapah Road, Singapore 487372}
\affiliation{Centre for Quantum Technologies, National University of Singapore, 3 Science Drive 2, Singapore 117543}
\affiliation{Horizon Quantum Computing, 79 Ayer Rajah Crescent, Singapore 139955}
\begin{abstract}
Machine learning has recently emerged as a fruitful area for finding potential quantum computational advantage.  
Many of the quantum enhanced machine learning algorithms critically hinge upon the ability to efficiently produce states proportional to high-dimensional data points stored in a quantum accessible memory. 
Even given query access to exponentially many entries stored in a database, the construction of which is considered a one-off overhead, it has been argued that the cost of preparing such amplitude-encoded states may offset any exponential quantum advantage.
Here we prove using smoothed analysis, that if the data-analysis algorithm is robust against small entry-wise input perturbation, state preparation can always be achieved with constant queries. This criterion is typically satisfied in realistic machine learning applications, where input data is subjective to moderate noise. Our results are equally applicable to the recent seminal progress in quantum-inspired algorithms, where specially constructed databases suffice for polylogarithmic classical algorithm in low-rank cases. The consequence of our finding is that for the purpose of practical machine learning, polylogarithmic processing time is possible under a general and flexible input model with quantum algorithms or quantum-inspired classical algorithms in the low-rank cases. \end{abstract}

\maketitle
\date{today}
\maketitle
In recent years, there has been substantial interest in algorithms based on ``quantum linear algebra'', where quantum states are used to represent vectors with exponentially large dimensions, which are manipulated by large matrices representing quantum operations \cite{zhao2019compiling, QSVT}. A particularly fruitful domain of applying such methods is quantum machine learning, where such algorithms promise for exponential and high-degree polynomial improvements \cite{aimeur2006machine,superunsuper,QVSM,QPC,PhysRevA.99.052331, 10.1088/1361-6633/aab406,zhao2019quantum,Zhao2019}. Most of these methods rely on quantum algorithms for solving linear systems that trace back to the seminal result of Harrow, Hassidim and Lloyd \cite{HHL}. This algorithm provides a way to generate quantum states proportional to $A^{-1} \mathbf{x}$, in time logarithmic in the dimension of the vector $\mathbf{x}$ which was subsequently improved to achieve exponentially better precision \cite{childs2017quantum}. However, as Aaronson pointed out \cite{Aaronson}, these quantum assisted machine learning approach comes with several caveats relating to the sparsity and conditioning of the matrix involved, and the structure of the input vector. Crucially, nearly all the algorithms assume access to a quantum accessible database which can produce quantum states with amplitudes proportional to the classical input entries, which is referred to as amplitude encoding.
This requirement persists even under the reasonable assumption that the loading and construction of the required databases constitute a one-off overhead, and does not contribute to the computational complexity analysis of data processing. 

A breakthrough came with the seminal work by Kerenidis and Prakash who provided the first end-to-end quantum machine learning algorithm with explicit description of quantum accessible data structure for efficient input preparation, which was applied to recommendation systems \cite{Kerenidis2016}. The same data structure was subsequently applied to improve the performance of the quantum linear system algorithm in dense cases \cite{wossnig2018quantum}. The Kerenidis-Prakash (KP) model explicitly stores the amplitudes in a binary tree structure. As such in addition to the ability of realising amplitude encoding for the row vectors of a matrix, the KP structure further allows for a stronger ability for preparing classical probability distributions proportional to vectors of entries squared -- so-called $\ell_{2}$-sampling. 
From a classical perspective, intriguingly, recent results of quantum-inspired machine learning algorithms \cite{tang2018quantum, gilyen2018quantum, TangPCA, chia2018quantum} have shown that assuming such a data structure capable of efficient $\ell_{2}$-sampling leads to equally efficient classical algorithms in the low-rank cases. The KP structure allows for input preparation both in the quantum and classical cases with a logarithmic storage overhead. The entries in the KP structure, i.e. the partial sums, are fixed beforehand and stored in the data structure. As a result, specific vectors for which the correct partial sums are stored can be directly accessed as amplitude encoded quantum state without post-selection. However, practical situations may arise where computation and re-computation of the partial sums is inefficient. For instance, it could be desirable to generate input vectors collecting entries from differing rows of a matrix stored in the data structure in statistical resampling techniques for cross-validation, or specific entry-wise functions are required for different algorithms using the same input data. In such situations, a more flexible and general input query model would have significant relevance.

In this letter, we exam the generic input model where only entry-wise access is allowed, and prove using smoothed analysis that both quantum amplitude encoding and classical $\ell_{2}$-sampling can be achieved with constant queries in realistic machine learning application with moderately noisy input data. In the quantum setting, entry-wise access is equivalent to having access to an oracle $O_{\mathbf{x}}$ that implements the unitary transformation,
\EQ{
\sum_{i,j} \alpha_{ij} \ket{i}\ket{j} \xrightarrow{O_{\mathbf{x}}} \sum_{i,j} \alpha_{ij} \ket{i}\ket{j + x_i}.\label{eqOracle}}
whereas in the classical setting, it is the standard random access memory that provides the mapping, $i \rightarrow x_i$. 
Generating the exact amplitude encoded states (or $\ell_{2}$-samples) from the entry-wise oracle, by the Grover (or unstructured) search lower bound requires $\Omega(D^{1/2})$ ($\Omega(D)$, respectively) calls. This observation is a common critique on early quantum machine proposals which did not explicitly address the issue of initial state preparation, yet attempted argue for polylogarithmic runtimes. Here we show that, in fact, for practical machine learning where the algorithm is robust against moderate input noise represented by entry-wise perturbations, quantum state preparation based on amplitude encoding has high success probability with only constant query cost. Furthermore, By an analogous argument, we also show that low-rank quantum-inspired algorithms operating under the same practical assumptions do not require access to special data structures to attain polylogarithmic runtime.

For quantum amplitude encoding, our objective is to prepare the state
\begin{align}
	\ket{\mathbf{x}} = \|\mathbf{x}\|_2^{-1}\sum_{i=1}^D x_i \ket{i}	
\end{align}
with amplitudes proportional to the entries $x_i$ of $\mathbf{x} \in \mathbb{R}^D$, given access to an oracle that realises the mapping in Eq. \ref{eqOracle}.
Quantum random access memory (QRAM) provides one way to implementing such an operation. In this case, $\mathbf{x}$ is stored classically and one is able to access the corresponding memory cells in quantum superposition \cite{RAM}. In other cases, the oracle can also be constructed if $x_i$ is efficiently computable given the index $i$.
	To produce $\ket{\mathbf{x}}$ probabilistically for any $\mathbf{x}$, one can start with the state $D^{-\frac{1}{2}}\sum_i \ket{i}\ket{0}$ as the query state for the operation in Eq.~(\ref{eqOracle}) to obtain $D^{-\frac{1}{2}} \sum_i \ket{i}\ket{x_i}$. An ancillary qubits is prepared in state $\ket{0}$ and then conditionally rotated based on the second register to obtain
 \begin{align}
 	 D^{-\frac{1}{2}} \sum_i \ket{i}\ket{x_i}\left(\sqrt{1-|x_i|^2} \ket{0} + x_i \ket{1}\right), \label{eq:rot}
 \end{align} 
 where we have assumed for simplicity that $\mathbf{x}$ is normalised such that $|x_i|\leq 1$. Performing a second oracle call to uncompute the registers $\ket{x_i}$, and post-selecting onto $\ket{1}$ results in the desired state $\ket{\mathbf{x}}$.

A known barrier to state preparation in the oracle setting is the probability of projecting onto the correct subspace in the final step, which is given by $D^{-1} \sum_i |x_i|^2$. When the entries of $\mathbf{x}$ are of similar magnitude, $\ket{\mathbf{x}}$ can be prepared using a constant number queries. However, in the case where a few entries are much larger than the rest, the lower bounds on unordered search \cite{boyer1996tight} imply that the corresponding state requires $\Omega\left(\sqrt{D}\right)$ oracle queries \cite{soklakov2006efficient}. This argument can be extended to the case where it is only necessary to prepare any $\ket{\mathbf{x'}}$ such that $|\mathbf{x'} - \mathbf{x}|_2$ is sufficiently small. 

However, here we make the observation that if real-valued data comes from a realistic source subjective to noise, and the algorithm is robust against such realistic input noise, then quantum state preparation can be done efficiently with constant oracle queries.
Formally, our argument is based on analysing the smoothed complexity for the amplitude encoding procedure when the input vectors are subjective to small perturbations. The smoothed complexity was introduced by Spielman and Teng \cite{spielman2009smoothed}, originally to explain the efficient performance of the simplex algorithm for linear programming in typical real-world scenarios. The same line of reasoning was subsequently applied to analyse the practical efficiency of various important algorithms in mathematical programming, machine learning, numerical analysis  discrete mathematics and combinatorics optimisation \cite{Spielman2001Smoothed}. 

The key intuition here is to analyse the performance of the algorithms when the worst-case data input is subjective to noise, which is represented by a small Gaussian element-wise perturbation. Following the convention of Ref. \cite{spielman2009smoothed}, we state the definition of smoothed complexity and then prove that preparing amplitude encoded states has a constant smoothed complexity:
\begin{definition}[Smoothed Complexity \cite{spielman2009smoothed}]
Given an algorithm $\mathcal{A}$ with an input domain $\Omega_D=\mathbb{R}^D$, the smoothed complexity of $\mathcal{A}$ with $\sigma$-Gaussian perturbation is defined as 
\begin{align}
	\text{Smoothed}^{\sigma}_{\mathcal{A}}(D)=\max_{\mathbf{x}\in [-1,1]^D} \mathbb{E}_{\mathbf{g}}[T_\mathcal{A}(\mathbf{x}+\mathbf{g})], \label{smooth}
\end{align}
where $\mathbf{g}$ is a Gaussian random vector with variance $\sigma^2$, and $T_\mathcal{A}$ denotes the runtime of $\mathcal{A}$.
\end{definition}
Furthermore, $\mathcal{A}$ is said to have polynomial smoothed complexity if there exist positive constants $k_1$, $k_2$, $D_0$, $\sigma_0$ and $c$ such that for all $D\ge D_0$ and $0\le \sigma \le \sigma_0$, it holds that 
\begin{align}
\text{Smoothed}^{\sigma}_{\mathcal{A}}(D)\le c \sigma^{-k_2} D^{k_1}.	
\end{align}
\begin{theorem}
	\label{AEsmoothed}
	Given oracle access, $O_\mathbf{x}$ to the entries of $\mathbf{x}\in \mathbb{R}^D$, the amplitude encoding of $\mathbf{x}$ into $\ket{\mathbf{x}}$ has smoothed complexity $\mathcal{O}(1/\sigma)$.
\end{theorem}
\begin{proof}
Let $\mathcal{A}$ be the algorithm that maps $D^{-\frac{1}{2}} \sum_i \ket{i}\ket{x_i}$ into $\ket{\mathbf{x}} = \|\mathbf{x}\|_2^{-1}\sum_{i=1}^D x_i \ket{i}$. After applying the controlled rotation and uncomputing the second register,  
the optimal success probability of projecting the state,
\begin{align}
D^{-\frac{1}{2}} \sum_i \ket{i}\left(\sqrt{1-|x_i|^2} \ket{0} + x_i \ket{1}\right)	
\end{align}
onto the desired state $\ket{\mathbf{x}}$ is given by $P_{\mathcal{A}}=\mathcal{O}(\frac{\|\mathbf{x}\|_2}{\sqrt{D}})$ with fixed-point amplitude amplification \cite{gilyen2018quantum}. From the definition of smoothed complexity, Eq.~\ref{smooth}, in the worst case, we have 
\begin{align}
\text{Smoothed}^{\sigma}_{\mathcal{A}}(D)=&\max_{\mathbf{x}\in [-1,1]^D} \mathbb{E}_{\mathbf{g}}[T_\mathcal{A}(\mathbf{x}+\mathbf{g})]	\nonumber\\
=& \left( \min_{\mathbf{x}\in [-1,1]^D} \mathbb{E}_{\mathbf{g}}[P_\mathcal{A}(\mathbf{x}+\mathbf{g})]\right)^{-1} \nonumber\\
=& \left(\mathbb{E}_{\mathbf{g}}[P_\mathcal{A}(\mathbf{g})]\right)^{-1}\nonumber\\
=& D^{\frac{1}{2}} \left(\mathbb{E}[\|\mathbf{g}\|_2]\right)^{-1}.\nonumber\\
=& \mathcal{O}\left(\frac{1}{\sigma}\right),\label{smoothresult}
\end{align} 	
where the last equality followed by noting that the random variable, $\|\mathbf{g}\|_2$, by definition follows a chi distribution with mean,
\begin{align}
\mathbb{E}[\|\mathbf{g}\|_2] = \sqrt{2}\sigma\frac{\Gamma((D+1)/2)}{\Gamma(D/2)}=\mathcal{O}(\sigma\sqrt{D} ).	
\end{align}
\end{proof}
\noindent The result of Eq.~\ref{smoothresult} implies that when the input vector is subjective to a certain level of noise represented by an element-wise Gaussian perturbation, the query complexity of preparing amplitude encoding is independent of the dimensionality.

We have seen that the quantum state preparation based on amplitude encoding has a constant runtime given that the input is subjective to a finite variance Gaussian noise. An analogous reasoning applies in the fully classical setting.
Classically, particularly in the context of quantum-inspired ML algorithm, $\ell_2$-sampling can be achieved by simple rejection sampling: an entry is chosen uniformly at random, and a value $x_j$ is read (we assume $x_j\leq c$, and $c$ is a known constant upper bound on the entries). Then a random real is sampled from the interval $(0 , c)$, and if this value is below $|x_j|^2$, the value $j$ is output.  Otherwise the process is repeated. The acceptance probability of the element $j$ is given by $|x_j|^2$ as desired, which leads to an average runtime of $D(\sum_j|x|_j^2)^{-1}$ for producing an $\ell_2$-sample from the correct distribution. We can make an analogous smoothed analysis for this classical rejection sampling process. Denoting $\mathcal{R}$ as the algorithm that performs rejection sampling, we have 
\begin{align}
\text{Smoothed}^{\sigma}_{\mathcal{R}}(D)=&\max_{\mathbf{x}\in [-1,1]^D} \mathbb{E}_{\mathbf{g}}[T_\mathcal{R}(\mathbf{x}+\mathbf{g})]	\nonumber\\
=& D \left(\mathbb{E}[\|\mathbf{g}\|_2]\right)^{-2}.\nonumber\\
=& \mathcal{O}\left(\frac{1}{\sigma^2}\right).
\end{align} 	
Thus the smoothed complexity of preparing quantum amplitude encoding from QRAM queries and classical $\ell_2$-samplings from classical RAM are $\mathcal{O}(\frac{1}{\sigma})$ and $\mathcal{O}(\frac{1}{\sigma^2})$ respectively given a $\sigma^2$-variance Gaussian perturbation on the input. The quadratic quantum improvement in the dependency on $\sigma$ comes directly from amplitude amplification.

In practical settings, the effect of a noisy element-wise perturbation on input data is well-studied in the machine learning literature \cite{Nettleton2010A, Kalapanidas2003Machine,Cesa2011Online}. As a specific example, in the domain of computer vision, Ref. \cite{dodge2016understanding}
examined the effect of various synthetic noise effects on the performance of popular deep learning
computer vision models such as Caffe, VGG16, VGG-CNN-S, and GoogLeNet.
Two relevant results
from the work, the pixel-wise Gaussian noise and a change in contrast to the image, make a little
effect on the performance. 
While adversarial attacks aim to find the worst
case corruption which leads to misclassification, Ref. \cite{dodge2016understanding} empirically shows that random
perturbations of a half pixel and small shifts in the pixel values will likely have a negligible effect. 
Furthermore, if the entry-wise perturbation represents a systematic shift in the data points, it has no effect on a large class of useful machine learning models. For examples, kernel methods such as Gaussian processes, support vector machines, determinantal point processes and Gaussian Markov random fields, to name but a few, most commonly use stationary kernels which are shift invariant \cite{genton2001classes}.
More generally, any digital data processing with floating-point arithmetic only makes sense if the overall results of such a computation maintain its validity when the features of the input vector had been perturbed below the machine precision. 
This is also practically reasonable, since real world data come from measurements of finite precision.

An element-wise perturbation can be represented by an off-set in the $\infty$-norm induced distance between the original vector $\mathbf{x}$ and the perturbed vector $\mathbf{x}'$. Assuming that the data analysis process, along with the particular instance of the input data, is robust against such small $\infty$-norm perturbations, we can in fact choose to work with the vector entries $x'_i$ which are half-integer multiples of the base precision $\epsilon$. Such that $\mathbf{x'}$ is chosen to be the closest representable vector to $\mathbf{x}$ (as shown in Figure \ref{fig:numbering}), which satisfies $|\mathbf{x'} - \mathbf{x}|_\infty\leq \frac{\epsilon}{2}$ and the distance from the true value of the data is less than $\epsilon$. This offset rounding can be implemented in the oracle access stage, or effectively realised at the controlled rotation stage, as discussed earlier. The aforementioned robustness assumption in data processing implies the analysis of results are insensitive to such offset rounding. In the quantum setting, the key benefit in using the offset rounding is that, since the exact representation of $0$ is not included in this offset rounding convention, the probability of the final projection step succeeding is at least $\frac{\epsilon^2}{4}$, and hence independent of $D$. This leads to an expected number of queries of $\mathcal{O} (1/\epsilon^2)$. The query complexity can be improved to $\mathcal{O}(1/\epsilon)$ using fixed-point quantum amplitude amplification based on the methods presented in Ref. \cite{gilyen2018quantum}, which achieves the same optimal query complexity as in Ref. \cite{yoder2014fixed}, but without introducing an additional phase that could be undesirable if multiple vectors are required to be prepared in superposition.
Note that in some cases, a systematic perturbation of data-vectors, by utilising, say, a positive sign offset ($+\epsilon/2$) to data-points is undesirable. 
In these cases, one can opt for a near-white noise offset using either a suitable pseudo-random number generator seeded by the memory location being queried or by adding random data $w_i$, performing $\ket{i}\ket{x_i} \to \ket{i}\ket{x_i + w_i}$. Critically, however, the number of oracle queries necessary to successfully prepare the state always has an upper bound that is independent from the database size.
\begin{figure}[t]
	\includegraphics[width = \columnwidth]{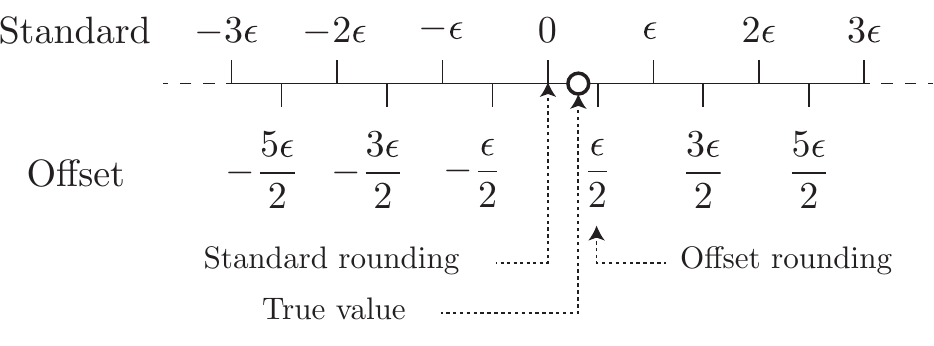}
	\caption{Numerical rounding conventions. In the standard rounding convention scalar values are rounded to the nearest integer multiple of precision $\epsilon$. Alternatively, we can consider an offset rounding convention, where the rounding is to the nearest half-integer multiple of $\epsilon$. As a result, this numbering convention does not contain an exact representation of $0$. In either scheme, the rounded value is always within $\frac{\epsilon}{2}$ of the true value. \label{fig:numbering}}
\end{figure}
It is worth cautioning that there exist scenarios, especially in computational learning theory \cite{arunachalam2017guest}, for which the robustness assumption against entry-wise data perturbation does not generally hold. For instance, if the input vector contains zeros with special meanings such as indicating a canonical vector or representing unknown values, the model may be sensitive toward shifting the zeros to even a small value. For instance, when loading a high-dimensional data point with constant or polylogarithmic sparsity, a systematic shift on the zero entries will produce a state vector converging to a uniform superposition, hence losing necessary information about the original input for meaningful analysis. Nevertheless, the constraint that $|\mathbf{x'} - \mathbf{x}|_\infty\leq \epsilon$, rather than requiring closeness in the 2-norm is still meaningful to a large class of practical machine learning tasks. We should also note that that the quantum state preparation discussed here is useful to quantum algorithms which aim to improve the efficiency of conventional classical algorithms but otherwise realise the same input-output functionality, such that the quantum algorithms inherit the desired robustness property of their conventional counter-parts. 

In summary, we have shown that any application which is robust under small $\infty$-norm perturbations, as it is the case in most practical machine learning, allows for efficient input preparation, both in the sense of classical $\ell_2$-sampling and the coherent amplitude encoded state preparation. In the context of quantum machine learning, this suggests that the caveat related to state preparation raised by Aaronson \cite{Aaronson} can generally be overcome for a wide range of practical use-cases, due to the natural robustness assumption. In the context of quantum-inspired machine learning, this finding removes the necessity of special data structures that involve the storage of partial sums. Hence we have provided a concrete argument for the feasibility of input preparation for both quantum and quantum-inspired algorithms for machine learning under the most general and flexible entry-wise query access model, which we believe will present both conceptual merit and practical utility to the promising exploration between quantum information and machine learning.

\textit{Acknowledgements--}The authors thank Ashley Montanaro and Ronald de Wolf for helpful comments on the manuscript. JFF and ZZ acknowledge support from Singapore's Ministry of Education and National Research Foundation. JFF acknowledges support from the Air Force Office of Scientific Research under AOARD grant FA2386-15-1-4082. This material is based on research funded in part by the Singapore National Research Foundation under NRF Award NRF-NRFF2013-01.

\textit{Competing financial interests--} Dr. Fitzsimons has financial holdings in Horizon Quantum Computing Pte Ltd. Otherwise, the authors have no potential financial or non-financial conflicts of interest.

\bibliographystyle{apsrev}
\bibliography{state_prep}

\begin{thebibliography}{31}
\expandafter\ifx\csname natexlab\endcsname\relax\def\natexlab#1{#1}\fi
\expandafter\ifx\csname bibnamefont\endcsname\relax
  \def\bibnamefont#1{#1}\fi
\expandafter\ifx\csname bibfnamefont\endcsname\relax
  \def\bibfnamefont#1{#1}\fi
\expandafter\ifx\csname citenamefont\endcsname\relax
  \def\citenamefont#1{#1}\fi
\expandafter\ifx\csname url\endcsname\relax
  \def\url#1{\texttt{#1}}\fi
\expandafter\ifx\csname urlprefix\endcsname\relax\def\urlprefix{URL }\fi
\providecommand{\bibinfo}[2]{#2}
\providecommand{\eprint}[2][]{\url{#2}}

\bibitem[{\citenamefont{Zhao et~al.}(2019{\natexlab{a}})\citenamefont{Zhao,
  Zhao, Rebentrost, and Fitzsimons}}]{zhao2019compiling}
\bibinfo{author}{\bibfnamefont{L.}~\bibnamefont{Zhao}},
  \bibinfo{author}{\bibfnamefont{Z.}~\bibnamefont{Zhao}},
  \bibinfo{author}{\bibfnamefont{P.}~\bibnamefont{Rebentrost}},
  \bibnamefont{and}
  \bibinfo{author}{\bibfnamefont{J.}~\bibnamefont{Fitzsimons}},
  \bibinfo{journal}{arXiv preprint arXiv:1902.10394}
  (\bibinfo{year}{2019}{\natexlab{a}}).

\bibitem[{\citenamefont{Gily\'{e}n et~al.}(2018)\citenamefont{Gily\'{e}n, Su,
  Low, and Wiebe}}]{QSVT}
\bibinfo{author}{\bibfnamefont{A.}~\bibnamefont{Gily\'{e}n}},
  \bibinfo{author}{\bibfnamefont{Y.}~\bibnamefont{Su}},
  \bibinfo{author}{\bibfnamefont{G.}~\bibnamefont{Low}}, \bibnamefont{and}
  \bibinfo{author}{\bibfnamefont{N.}~\bibnamefont{Wiebe}},
  \bibinfo{journal}{arXiv preprint arXiv:1806.01838}  (\bibinfo{year}{2018}).

\bibitem[{\citenamefont{A{\"\i}meur et~al.}(2006)\citenamefont{A{\"\i}meur,
  Brassard, and Gambs}}]{aimeur2006machine}
\bibinfo{author}{\bibfnamefont{E.}~\bibnamefont{A{\"\i}meur}},
  \bibinfo{author}{\bibfnamefont{G.}~\bibnamefont{Brassard}}, \bibnamefont{and}
  \bibinfo{author}{\bibfnamefont{S.}~\bibnamefont{Gambs}}, in
  \emph{\bibinfo{booktitle}{Advances in Artificial Intelligence}}
  (\bibinfo{publisher}{Springer}, \bibinfo{year}{2006}), pp.
  \bibinfo{pages}{431--442}.

\bibitem[{\citenamefont{Lloyd et~al.}(2013)\citenamefont{Lloyd, Mohseni, and
  Rebentrost}}]{superunsuper}
\bibinfo{author}{\bibfnamefont{S.}~\bibnamefont{Lloyd}},
  \bibinfo{author}{\bibfnamefont{M.}~\bibnamefont{Mohseni}}, \bibnamefont{and}
  \bibinfo{author}{\bibfnamefont{P.}~\bibnamefont{Rebentrost}},
  \bibinfo{journal}{arXiv preprint arXiv:1307.0411}  (\bibinfo{year}{2013}).

\bibitem[{\citenamefont{Rebentrost et~al.}(2014)\citenamefont{Rebentrost,
  Mohseni, and Lloyd}}]{QVSM}
\bibinfo{author}{\bibfnamefont{P.}~\bibnamefont{Rebentrost}},
  \bibinfo{author}{\bibfnamefont{M.}~\bibnamefont{Mohseni}}, \bibnamefont{and}
  \bibinfo{author}{\bibfnamefont{S.}~\bibnamefont{Lloyd}},
  \bibinfo{journal}{Physical Review Letters} \textbf{\bibinfo{volume}{113}},
  \bibinfo{pages}{130503} (\bibinfo{year}{2014}).

\bibitem[{\citenamefont{Schuld et~al.}(2014)\citenamefont{Schuld, Sinayskiy,
  and Petruccione}}]{QPC}
\bibinfo{author}{\bibfnamefont{M.}~\bibnamefont{Schuld}},
  \bibinfo{author}{\bibfnamefont{I.}~\bibnamefont{Sinayskiy}},
  \bibnamefont{and}
  \bibinfo{author}{\bibfnamefont{F.}~\bibnamefont{Petruccione}}, in
  \emph{\bibinfo{booktitle}{PRICAI 2014: Trends in Artificial Intelligence}}
  (\bibinfo{publisher}{Springer}, \bibinfo{year}{2014}), pp.
  \bibinfo{pages}{208--220}.

\bibitem[{\citenamefont{Zhao et~al.}(2019{\natexlab{b}})\citenamefont{Zhao,
  Fitzsimons, and Fitzsimons}}]{PhysRevA.99.052331}
\bibinfo{author}{\bibfnamefont{Z.}~\bibnamefont{Zhao}},
  \bibinfo{author}{\bibfnamefont{J.~K.} \bibnamefont{Fitzsimons}},
  \bibnamefont{and} \bibinfo{author}{\bibfnamefont{J.~F.}
  \bibnamefont{Fitzsimons}}, \bibinfo{journal}{Phys. Rev. A}
  \textbf{\bibinfo{volume}{99}}, \bibinfo{pages}{052331}
  (\bibinfo{year}{2019}{\natexlab{b}}).

\bibitem[{\citenamefont{Dunjko and Briegel}(2018)}]{10.1088/1361-6633/aab406}
\bibinfo{author}{\bibfnamefont{V.}~\bibnamefont{Dunjko}} \bibnamefont{and}
  \bibinfo{author}{\bibfnamefont{H.~J.} \bibnamefont{Briegel}},
  \bibinfo{journal}{Reports on Progress in Physics}  (\bibinfo{year}{2018}).

\bibitem[{\citenamefont{Zhao et~al.}(2019{\natexlab{c}})\citenamefont{Zhao,
  Fitzsimons, Osborne, Roberts, and Fitzsimons}}]{zhao2019quantum}
\bibinfo{author}{\bibfnamefont{Z.}~\bibnamefont{Zhao}},
  \bibinfo{author}{\bibfnamefont{J.~K.} \bibnamefont{Fitzsimons}},
  \bibinfo{author}{\bibfnamefont{M.~A.} \bibnamefont{Osborne}},
  \bibinfo{author}{\bibfnamefont{S.~J.} \bibnamefont{Roberts}},
  \bibnamefont{and} \bibinfo{author}{\bibfnamefont{J.~F.}
  \bibnamefont{Fitzsimons}}, \bibinfo{journal}{Physical Review A}
  \textbf{\bibinfo{volume}{100}}, \bibinfo{pages}{012304}
  (\bibinfo{year}{2019}{\natexlab{c}}).

\bibitem[{\citenamefont{Zhao et~al.}(2019{\natexlab{d}})\citenamefont{Zhao,
  Pozas-Kerstjens, Rebentrost, and Wittek}}]{Zhao2019}
\bibinfo{author}{\bibfnamefont{Z.}~\bibnamefont{Zhao}},
  \bibinfo{author}{\bibfnamefont{A.}~\bibnamefont{Pozas-Kerstjens}},
  \bibinfo{author}{\bibfnamefont{P.}~\bibnamefont{Rebentrost}},
  \bibnamefont{and} \bibinfo{author}{\bibfnamefont{P.}~\bibnamefont{Wittek}},
  \bibinfo{journal}{Quantum Machine Intelligence}
  (\bibinfo{year}{2019}{\natexlab{d}}), ISSN \bibinfo{issn}{2524-4914}.

\bibitem[{\citenamefont{Harrow et~al.}(2009)\citenamefont{Harrow, Hassidim, and
  Lloyd}}]{HHL}
\bibinfo{author}{\bibfnamefont{A.~W.} \bibnamefont{Harrow}},
  \bibinfo{author}{\bibfnamefont{A.}~\bibnamefont{Hassidim}}, \bibnamefont{and}
  \bibinfo{author}{\bibfnamefont{S.}~\bibnamefont{Lloyd}},
  \bibinfo{journal}{Physical Review Letters} \textbf{\bibinfo{volume}{103}},
  \bibinfo{pages}{150502} (\bibinfo{year}{2009}).

\bibitem[{\citenamefont{Childs et~al.}(2017)\citenamefont{Childs, Kothari, and
  Somma}}]{childs2017quantum}
\bibinfo{author}{\bibfnamefont{A.~M.} \bibnamefont{Childs}},
  \bibinfo{author}{\bibfnamefont{R.}~\bibnamefont{Kothari}}, \bibnamefont{and}
  \bibinfo{author}{\bibfnamefont{R.~D.} \bibnamefont{Somma}},
  \bibinfo{journal}{SIAM Journal on Computing} \textbf{\bibinfo{volume}{46}},
  \bibinfo{pages}{1920} (\bibinfo{year}{2017}).

\bibitem[{\citenamefont{Aaronson}(2015)}]{Aaronson}
\bibinfo{author}{\bibfnamefont{S.}~\bibnamefont{Aaronson}},
  \bibinfo{journal}{Nature Physics} \textbf{\bibinfo{volume}{11}},
  \bibinfo{pages}{291} (\bibinfo{year}{2015}).

\bibitem[{\citenamefont{Kerenidis and Prakash}(2017)}]{Kerenidis2016}
\bibinfo{author}{\bibfnamefont{I.}~\bibnamefont{Kerenidis}} \bibnamefont{and}
  \bibinfo{author}{\bibfnamefont{A.}~\bibnamefont{Prakash}}, in
  \emph{\bibinfo{booktitle}{Innovations in Theoretical Computer Science}}
  (\bibinfo{year}{2017}).

\bibitem[{\citenamefont{Wossnig et~al.}(2018)\citenamefont{Wossnig, Zhao, and
  Prakash}}]{wossnig2018quantum}
\bibinfo{author}{\bibfnamefont{L.}~\bibnamefont{Wossnig}},
  \bibinfo{author}{\bibfnamefont{Z.}~\bibnamefont{Zhao}}, \bibnamefont{and}
  \bibinfo{author}{\bibfnamefont{A.}~\bibnamefont{Prakash}},
  \bibinfo{journal}{Physical review letters} \textbf{\bibinfo{volume}{120}},
  \bibinfo{pages}{050502} (\bibinfo{year}{2018}).

\bibitem[{\citenamefont{Tang}(2018{\natexlab{a}})}]{tang2018quantum}
\bibinfo{author}{\bibfnamefont{E.}~\bibnamefont{Tang}}, \bibinfo{journal}{arXiv
  preprint arXiv:1807.04271}  (\bibinfo{year}{2018}{\natexlab{a}}).

\bibitem[{\citenamefont{Gily{\'e}n et~al.}(2018)\citenamefont{Gily{\'e}n,
  Lloyd, and Tang}}]{gilyen2018quantum}
\bibinfo{author}{\bibfnamefont{A.}~\bibnamefont{Gily{\'e}n}},
  \bibinfo{author}{\bibfnamefont{S.}~\bibnamefont{Lloyd}}, \bibnamefont{and}
  \bibinfo{author}{\bibfnamefont{E.}~\bibnamefont{Tang}},
  \bibinfo{journal}{arXiv preprint arXiv:1811.04909}  (\bibinfo{year}{2018}).

\bibitem[{\citenamefont{Tang}(2018{\natexlab{b}})}]{TangPCA}
\bibinfo{author}{\bibfnamefont{E.}~\bibnamefont{Tang}}, \bibinfo{journal}{arXiv
  preprint arXiv:1811.00414}  (\bibinfo{year}{2018}{\natexlab{b}}).

\bibitem[{\citenamefont{Chia et~al.}(2018)\citenamefont{Chia, Lin, and
  Wang}}]{chia2018quantum}
\bibinfo{author}{\bibfnamefont{N.-H.} \bibnamefont{Chia}},
  \bibinfo{author}{\bibfnamefont{H.-H.} \bibnamefont{Lin}}, \bibnamefont{and}
  \bibinfo{author}{\bibfnamefont{C.}~\bibnamefont{Wang}},
  \bibinfo{journal}{arXiv preprint arXiv:1811.04852}  (\bibinfo{year}{2018}).

\bibitem[{\citenamefont{Giovannetti et~al.}(2008)\citenamefont{Giovannetti,
  Lloyd, and Maccone}}]{RAM}
\bibinfo{author}{\bibfnamefont{V.}~\bibnamefont{Giovannetti}},
  \bibinfo{author}{\bibfnamefont{S.}~\bibnamefont{Lloyd}}, \bibnamefont{and}
  \bibinfo{author}{\bibfnamefont{L.}~\bibnamefont{Maccone}},
  \bibinfo{journal}{Physical Review Letters} \textbf{\bibinfo{volume}{100}},
  \bibinfo{pages}{160501} (\bibinfo{year}{2008}).

\bibitem[{\citenamefont{Boyer et~al.}(1996)\citenamefont{Boyer, Brassard,
  H{\o}yer, and Tapp}}]{boyer1996tight}
\bibinfo{author}{\bibfnamefont{M.}~\bibnamefont{Boyer}},
  \bibinfo{author}{\bibfnamefont{G.}~\bibnamefont{Brassard}},
  \bibinfo{author}{\bibfnamefont{P.}~\bibnamefont{H{\o}yer}}, \bibnamefont{and}
  \bibinfo{author}{\bibfnamefont{A.}~\bibnamefont{Tapp}},
  \bibinfo{journal}{arXiv preprint quant-ph/9605034}  (\bibinfo{year}{1996}).

\bibitem[{\citenamefont{Soklakov and Schack}(2006)}]{soklakov2006efficient}
\bibinfo{author}{\bibfnamefont{A.~N.} \bibnamefont{Soklakov}} \bibnamefont{and}
  \bibinfo{author}{\bibfnamefont{R.}~\bibnamefont{Schack}},
  \bibinfo{journal}{Physical Review A} \textbf{\bibinfo{volume}{73}},
  \bibinfo{pages}{012307} (\bibinfo{year}{2006}).

\bibitem[{\citenamefont{Spielman and Teng}(2009)}]{spielman2009smoothed}
\bibinfo{author}{\bibfnamefont{D.~A.} \bibnamefont{Spielman}} \bibnamefont{and}
  \bibinfo{author}{\bibfnamefont{S.-H.} \bibnamefont{Teng}},
  \bibinfo{journal}{Communications of the ACM} \textbf{\bibinfo{volume}{52}},
  \bibinfo{pages}{76} (\bibinfo{year}{2009}).

\bibitem[{\citenamefont{Spielman and Teng}(2001)}]{Spielman2001Smoothed}
\bibinfo{author}{\bibfnamefont{D.}~\bibnamefont{Spielman}} \bibnamefont{and}
  \bibinfo{author}{\bibfnamefont{S.~H.} \bibnamefont{Teng}}, in
  \emph{\bibinfo{booktitle}{ACM Symposium on Theory of Computing}}
  (\bibinfo{year}{2001}), pp. \bibinfo{pages}{296--305}.

\bibitem[{\citenamefont{Nettleton et~al.}(2010)\citenamefont{Nettleton,
  Orriols-Puig, and Fornells}}]{Nettleton2010A}
\bibinfo{author}{\bibfnamefont{D.~F.} \bibnamefont{Nettleton}},
  \bibinfo{author}{\bibfnamefont{A.}~\bibnamefont{Orriols-Puig}},
  \bibnamefont{and} \bibinfo{author}{\bibfnamefont{A.}~\bibnamefont{Fornells}},
  \bibinfo{journal}{Artificial Intelligence Review}
  \textbf{\bibinfo{volume}{33}}, \bibinfo{pages}{275} (\bibinfo{year}{2010}).

\bibitem[{\citenamefont{Kalapanidas et~al.}(2003)\citenamefont{Kalapanidas,
  Avouris, and Craciun}}]{Kalapanidas2003Machine}
\bibinfo{author}{\bibfnamefont{E.}~\bibnamefont{Kalapanidas}},
  \bibinfo{author}{\bibfnamefont{N.}~\bibnamefont{Avouris}}, \bibnamefont{and}
  \bibinfo{author}{\bibfnamefont{M.}~\bibnamefont{Craciun}}
  (\bibinfo{year}{2003}).

\bibitem[{\citenamefont{Cesa-Bianchi et~al.}(2011)\citenamefont{Cesa-Bianchi,
  Shalev-Shwartz, and Shamir}}]{Cesa2011Online}
\bibinfo{author}{\bibfnamefont{N.}~\bibnamefont{Cesa-Bianchi}},
  \bibinfo{author}{\bibfnamefont{S.}~\bibnamefont{Shalev-Shwartz}},
  \bibnamefont{and} \bibinfo{author}{\bibfnamefont{O.}~\bibnamefont{Shamir}},
  \bibinfo{journal}{IEEE Transactions on Information Theory}
  \textbf{\bibinfo{volume}{57}}, \bibinfo{pages}{7907} (\bibinfo{year}{2011}).

\bibitem[{\citenamefont{Dodge and Karam}(2016)}]{dodge2016understanding}
\bibinfo{author}{\bibfnamefont{S.}~\bibnamefont{Dodge}} \bibnamefont{and}
  \bibinfo{author}{\bibfnamefont{L.}~\bibnamefont{Karam}}, in
  \emph{\bibinfo{booktitle}{Quality of Multimedia Experience (QoMEX), 2016
  Eighth International Conference on}} (\bibinfo{organization}{IEEE},
  \bibinfo{year}{2016}), pp. \bibinfo{pages}{1--6}.

\bibitem[{\citenamefont{Genton}(2001)}]{genton2001classes}
\bibinfo{author}{\bibfnamefont{M.~G.} \bibnamefont{Genton}},
  \bibinfo{journal}{Journal of machine learning research}
  \textbf{\bibinfo{volume}{2}}, \bibinfo{pages}{299} (\bibinfo{year}{2001}).

\bibitem[{\citenamefont{Yoder et~al.}(2014)\citenamefont{Yoder, Low, and
  Chuang}}]{yoder2014fixed}
\bibinfo{author}{\bibfnamefont{T.~J.} \bibnamefont{Yoder}},
  \bibinfo{author}{\bibfnamefont{G.~H.} \bibnamefont{Low}}, \bibnamefont{and}
  \bibinfo{author}{\bibfnamefont{I.~L.} \bibnamefont{Chuang}},
  \bibinfo{journal}{Physical review letters} \textbf{\bibinfo{volume}{113}},
  \bibinfo{pages}{210501} (\bibinfo{year}{2014}).

\bibitem[{\citenamefont{Arunachalam and de~Wolf}(2017)}]{arunachalam2017guest}
\bibinfo{author}{\bibfnamefont{S.}~\bibnamefont{Arunachalam}} \bibnamefont{and}
  \bibinfo{author}{\bibfnamefont{R.}~\bibnamefont{de~Wolf}},
  \bibinfo{journal}{ACM SIGACT News} \textbf{\bibinfo{volume}{48}},
  \bibinfo{pages}{41} (\bibinfo{year}{2017}).

\end{thebibliography}
\end{document}